\documentclass{llncs}
\usepackage{amsmath,amssymb,amstext,amsfonts,graphics, enumerate}
\usepackage{color,graphicx}

\def\R{{\mathbb R}}

\begin{document}
\title{Note on VCG vs. Price Raising\\ for Matching Markets}

\author{Walter Kern, Bodo Manthey, Marc Uetz}
\institute{University of Twente, P.O.Box 217, NL-7500 AE Enschede\\
\texttt{w.kern@utwente.nl}
}

\maketitle

\begin{abstract}
In \cite{EK10} the use of VCG in matching markets is motivated by saying that in order to compute market clearing prices in a matching market, the auctioneer needs to know the true valuations of the bidders. Hence VCG and corresponding 
personalized prices are proposed as an incentive compatible mechanism. The same line of argument pops up in several lecture sheets and other documents related to courses based on Easley and Kleinberg's book, seeming to suggest that computing
market clearing prices and corresponding assignments were \emph{not} incentive compatible. Main purpose of our note is to observe that, in contrast, assignments based on buyer optimal market clearing prices are indeed incentive compatible.\footnote{
In \cite{EK10}(section 15.9) it is shown that the VCG prices paid by the buyers are market clearing, but the conclusion that, conversely, buyer optimal market clearing prices are incenticve compatible, is missing.}
\end{abstract}

{\bf keywords:} matching markets, market clearing prices, VCG mechanism.

\section{Introduction}\label{s-intro}

The idea of market clearing prices seems to exist since more than 200 years (\emph{cf. \cite{Say34})}: Prices of products are expected to adapt to the demand such that the demand is satisfied without any leftover products. A rigorous mathematical
treatment in the special (and particularly simple) case of a ``two-sided market'' is 
presented in Shapley and Shubik's paper \cite{SS72}. The setup there consists of a set $I$ of \emph{items} (or ``sellers''), a set $J$ of \emph{buyers} and
given \emph{valuations} $v_{ij} \ge 0$ indicating the value that item $i$ has for buyer $j$. One may assume \emph{w.l.o.g.} that $|I|=|J|$ (otherwise add dummy items of value $0$ for all buyers or dummy buyers with zero valuations for all items
if necessary). Relative to given prices $p_i, i \in I$, buyer $j$ will only be interested in items that yield him a maximum profit $v_{ij}-p_i$. Let us call these items ``acceptable'' (for buyer $j$ \emph{w.r.t.} given prices $p_i$). The price vector
$p= (p_i)$ is \emph{market clearing} if there is a perfect matching of the buyers with the items such that each buyer is
assigned to (matched with) an acceptable item.

Market clearing prices are easy to compute (\emph{cf.} \cite{EK10} or section \ref{s-clearing}), even if we impose the
additional requirement that $\sum_i p_i$ should be minimal. Such \emph{buyer optimal} market clearing prices turn out to be
unique (\emph{cf.} \cite{SS72} or section \ref{s-clearing}).

The above setup is also interesting in the context of \emph{sealed bid auctions}, where buyers are called \emph{bidders}:
Assume that an \emph{auctioneer} is to assign
the items $i \in I$ to the buyers $j \in J$ and compute prices $p_i$ to be paid by the buyers for the items they get assigned to. The auctioneer employs a \emph{mechanism} performing this job (computing an assignment and corresponding prices 
based on the reported valuations $v_{ij}, i \in I, j \in J$). 
Such a mechanism  is called \emph{incentive compatible} if no bidder has any reason to lie about his valuations, \emph{i.e.} if telling his true valuations $v_{ij}, i \in I$ maximizes the profit $v_{ij}-p_i$ of bidder $j$ 
(where $i$ is the item he gets assigned to), irrespective of what the other bidders do. In game theoretic terms one would say that truthtelling
is a \emph{dominant strategy} for each bidder $j$.
A famous mechanism with this property is the so-called \emph{VCG mechanism}, based on work by  to Vickrey \cite{V61}, Clarke \cite{C71} and Groves \cite{G73}. The basic idea is that each bidder should pay for ``the harm it does to the others'' 
(\emph{cf.} section \ref{s-VCG} for details).\\

In \cite{EK10}, the use of the VCG-mechanism for matching markets (in particular sponsored search) is motivated by 
saying that computation of market clearing prices ``can only be carried out by a search engine that actually knows the valuations'' $v_{ij}$ of the bidders.  VCG is known to be incentive compatible (even in a more general setting,
\emph{cf.} \cite {N07}), so bidders can indeed be assumed to report their true valuations, once they know that VCG is applied 
to compute the assignment and corresponding prices.\\

As it turns out, however, these VCG prices are exactly the buyer optimal prices $p_i$ which we could equally well compute by the standard price raising procedure (\emph{cf.} \cite{EK10} or section \ref{s-clearing}) 
and subsequent reduction to buyer optimal prices. 
The purpose of our note is to stress that - at least in the context of matching markets - VCG is just an alternative way of defining buyer optimal prices, or, to put it in another way, a means of showing
that buyer optimal prices are incentive compatible.\\

We do not claim that our findings (in particular Lemma \ref{buyeropt=VCG}) are completely new, but on the other hand we also did could not find them anywhere
mentioned explicitly in the literature. Actually, we only found one single document pointing out an equivalence between VCG and market
clearing prices (\cite{A13}) and this completely ignores the crucial  buyer optimality. All others adopted
the argument from Easley and Kleinberg's book mentioned above.\\

In the following two sections we discuss market clearing prices, buyer optimality, and the equivalence with VCG in the context of matching markets. We aimed at a mainly selfcontained presentation (assuming only basic knowledge on
linear programming and bipartite matchings).

\section{Market clearing prices}\label{s-clearing}
Let $G=(I\cup J, E)$ be the complete bipartite graph with node sets $I$ (\emph{items}) and $J$ (\emph{buyers}),  and assume \emph{w.l.o.g.} that $|I|=|J|=n$, say. (Otherwise, add \emph{dummy} items of value $0$ to all players or 
\emph{dummy} buyers with valuation $0$ for all items.) Let $v_{ij} \ge 0, i \in I, j\in J$ denote the valuation of buyer $j \in J $ for item $i\in I$. A \emph{matching} in $G$ is a subset $M \subseteq E$ of paiwise non-incident edges. $M$ is 
\emph{perfect} if it matches all nodes, \emph{i.e.,} $|M|=n$. The \emph{value} of a matching $M$ is defined as
$v(M)= \sum_{ij \in M} v_{ij}$. We let 
$$v^* =v^*(G) = \max\{v(M)~|~ M \text{is perfect matching in } G\}.$$

It is well-known that $v^*$ and a corresponding optimal matching $M^* \subseteq E$ can be found by solving the following linear program $(P)$ with variables $x_{ij}$ indicating whether $ij \in M (x_{ij}=1)$ or not $(x_{ij}=0$ and its dual $(D)$:
\begin{eqnarray}
(P)~~~~~~~v^*= \max & \sum_{ij \in E} v_{ij}x_{ij} ~~~~~~~~~~~= ~~~~~ \min &\sum_i p_i +\sum_j q_j ~~~~~~~(D) \nonumber \\
& \sum_j x_{ij}=1, ~~i \in I~~~~~~~~~~ &~~ p_i+q_j \ge v_{ij}\nonumber\\
& \sum_i x_{ij}=1, ~~j \in J~~~~~~~~~~ &\nonumber
\end{eqnarray}

The dual variable $p_i$ may be interpreted as \emph{price} of item $i$ (see also below). To each price vector $p \in \R^I$ there is a corresponding best choice for $q$, obtained by setting $q_j:= \max_i v_{ij}-p_i$. For this reason we
also refer to a given price vector $p$ as a ``dual solution'' (without mentioning the corresponding $q$). Clearly, 
we may - and will - always assume $p \ge 0$ \emph{w.l.o.g.} (by adding a suitable amount to all components of $p$ and subtracting the same amount from $q$).\\

The \emph{price raising} procedure (\emph{cf., e.g.,} \cite{EK10}), in discrete optimization also known as \emph{Hungarian Method} (\cite{Kuhn51}, \emph{cf.} also \cite{Hi41}) is a rather efficient combinatorial algorithm for solving $(P)$ and $(D)$. 
It is a special case of a primal-dual algorithm that subsequently modifies a feasible dual solution $p$ until a complementary (hence optimal) primal solution is found:\\

Given a feasible dual solution $(p,q)$, define the corresponding set of \emph{equality edges}
$$E^= := \{e=ij \in E~|~ p_i+q_j=v_{ij}\}.$$
By complementary slackness, $(p,q)$ is optimal if there is a corresponding perfect matching $M \subseteq E^=$, 
\emph{i.e.,} a $0-1$ solution of $(P)$ with $x_{ij}=1 \Rightarrow p_i+q_j=v_{ij}$. Interpreting $p_i$ as price of item
$i$, such a matching $M$ assigns each buyer $j$ to an item $i$ with profit $q_j=v_{ij}-p_i$ for buyer $j$. Dual feasibility ensures that any other item $\tilde{i}$ would result in a lower or equal profit $v_{\tilde{i}j}-p_{\tilde{i}}$.
(This is why $G^==(I\cup J, E^=)$ is called ``best sellers graph'' sometimes (\emph{cf., e.g.,} \cite{EK10}).\\

Starting, say, from $p=0$, we successively seek to modify the current dual solution $(p,q)$ such that $G^= (I\cup J, E^=)$ contains larger and larger matchings $M \subseteq E^=$ in each iteration, until eventually a perfect matching
is reached (and hence an optimal solution is found).\\

Thus assume $(p,q)$ is a current dual feasible solution and $M \subseteq E^=$ is a maximum (non-perfect) matching, \emph{i.e.,} there exists an unmatched $j \in J$. Let $R\subseteq I$ denote the set of nodes that can be reached from $j$ along an $M$-alternating path in $G^=$ (as indicated in the figure below).\\[2mm]

 \begin{figure}\label{f-fig1} \includegraphics[height=20mm,width=30mm]{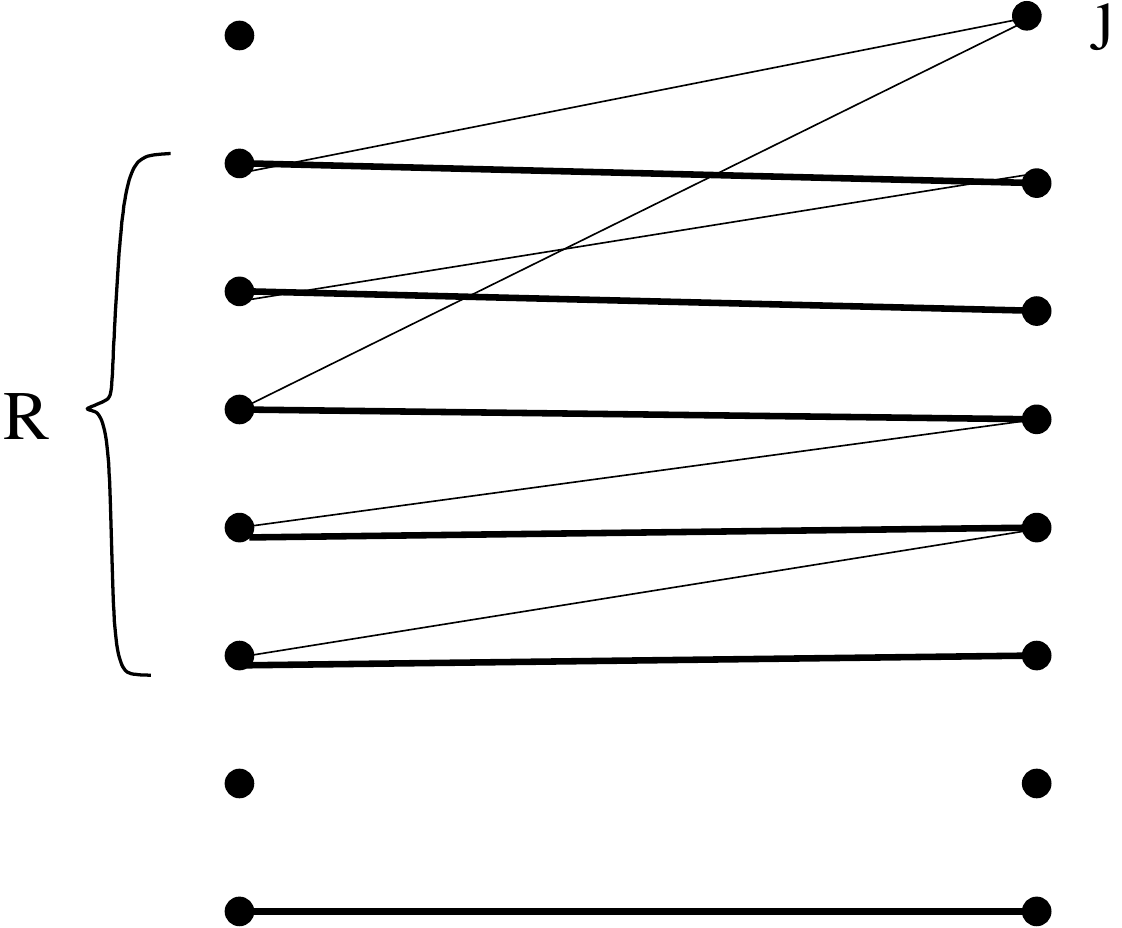}\end{figure}

Each $i \in R$ is matched by $M$ (otherwise the $M$-alternating path from $j$ to $i$ were $M$-augmenting, contradicting our assumption that $M$ is maximum). Moreover, the matching partners $M_R \subseteq J$ of $R$ are not joined to any
node $i \in I\backslash R$ in $G^=$ (any such $i$ would belong to $R$ by definition). Hence we may increase $p$ on $R$ and decrease $q$ on $M_R \cup \{j\}$, maintaining dual feasibility, until some $e \in E$ joining $I \backslash R$ to $M_R$ becomes an equality edge and we update (enlarge $R$ correspondingly. After at most $n=|I|$ steps, a new dual feasible $(\tilde{p}, \tilde{q})$ will be found, allowing a larger ``primal'' matching $\tilde{M}$ of size $|\tilde{M}|=|M|+1$ in the corresponding $\tilde{G}^=$.\\

After at most $n=|I|$ such iterations, we eventually end up with a dual feasible $(p,q)$ and a primal perfect matching
$M \subseteq E^=$ in the corresponding $G^=$. As explained earlier, given prices $p_i$, buyer $j$ is ready to buy any
 item $i$ with $ij \in E^=$ as this will maximize his profit. This explains the notion of \emph{market clearing} prices
for the optimal dual solutions on $D$ (as these are exactly the ones that allow a complementary perfect matching $M
\subseteq E^=$).\\

Market clearing prices are by no means unique. For example, raising all prices by $1$ (and decreasing $q$ correspondingly) does not affect optimality. However, there is a unique optimal dual solution $p$ that is \emph{buyer optimal} in that
it minimizes $\sum_i p_i$ (recall that we assume all prices to be non-negative). This is an immediate consequence of 
Shapley and Shubik's \cite{SS72} ``lattice structure'' Theorem:
\begin{theorem}\label{ShSh}
If $p',p'' \in \R^n_+$ are market clearing, so is $p_i := \min\{p'_i,p''_i\}, i \in I$.
\end{theorem}
\begin{proof}
It is straightforward to verify that $p$, together with $q_j := \max\{q'_j,q''_j\}$  yields a feasible dual solution.
(Here $q',q''$ are the dual variables corresponding to $p',p''$.) Furthermore, if $M$ is any maximum valuation perfect matching, then $M$ must be complementary with both $(p',q')$ and $p'',q'')$. Thus $p'_i+q'_j=p''_i+q''_j=v_{ij}$ for all
$ij \in M$. Hence $p_i+q_j=\min \{p'_i,p''_i\}+\max \{q'_j,q''_j\}=v_{ij}$ for all $ij \in M$, proving that $p$ is market
clearing as well. \qed
\end{proof}

Once we have computed market clearing prices $p$ by price raising as described above, we may further reduce them to
 unique buyer optimal prices in a similar way. The following observation is useful in this context (one direction is
proved in \cite{EK10}):

\begin{lemma}\label{buyeropt}
Let $p\in R^n_+$ be market clearing and $M \subseteq E^=$ be a corresponding perfect matching in the corresponding 
equality graph $G^=$. Then the following are equivalent:\\
$(i)$ $p$ is buyer optimal\\
$(ii)$ Each $j\in J$ is joined by an $M$-alternating path $P \subseteq E^=$, starting and ending with a matching edge,
to an item $i \in I$ with $p_i=0$.
\end{lemma}
\begin{proof}
``$\Rightarrow$'': Assume $p$ is buyer optimal and $j\in J$. Clearly, $p_{\min}= \min p_i = 0$ -- otherwise $p$ cannot be
buyer optimal (decrease all $p_i$ by $p_{\min}$). Let $R \subseteq I$ denote the set of items $i\in I$ that can be reached from $j$ along $M$-alternating paths $P \subseteq E^=$ from $j$, starting (and ending) with a matching edge. By 
definition of $R$, there is no edge in $E^=$ joining $R$ to any $j \in J\backslash M_R$, where again $M_R \subseteq J$ denotes the set of matching partners of $R$. Hence, in case $p >0$ on $R$, we may decrease $p$ on $R$ and increase $q$ on 
 $M_R$ while staying dually feasible and complementary with $M$. This would contradict the assumed minimality of $p$.\\
``$\Leftarrow$'': Assume that $(ii)$ holds, but $p$ is not buyer optimal, \emph{i.e.,} there exists a market clearing 
$\tilde{p}$ with $\tilde{p}_{i_0} <p_{i_0}$. Let $j_0$ be the matching partner of $i_0$ and let $P= j_0,i_0,j_1,i_1, ..., j_k,i_k$ be the $M$-alternating path whose existence is guaranteed by $(ii)$. As $\tilde{p}$ with corresponding $\tilde{q}$ is an optimal dual solution, complementarity with (the optimal primal) $M$ implies 
$$\tilde{p}_{i_0}+\tilde{q}_{j_0}=v_{i_0,j_0},~ ...,~ \tilde{p}_{i_k}+\tilde{q}_{j_k}=v_{i_k,j_k}.$$
Hence dual feasibility of $(\tilde{p},\tilde{q})$ yields
$$\tilde{p}_{i_0} \le p_i-\epsilon \Rightarrow \tilde{q}_{j_1} \ge q_i+\epsilon \Rightarrow \tilde{p}_{i_1} \le p_{i_1} - \epsilon \Rightarrow ... \Rightarrow \tilde{p}_{i_k} \le p_{i_k} -\epsilon = -\epsilon <0,$$
a contradiction. \qed
\end{proof}

Lemma \ref{buyeropt} suggests the following algorithm to reduce a given market clearing $p$ (with corresponding $q$ and any fixed maximum valuation matching $M$) towards buyer optimality: First decrease all prices uniformly until $p_i=0$ for some item $i$.  Check whether 
condition $(ii)$ in Lemma \ref{buyeropt} is satisfied. If yes, we are done. Otherwise pick any $j \in J$ that violates 
condition $(ii)$ and let $R \subseteq I$ denote the set of items that can be reached from $j$ along $M$-alternating 
paths starting (and ending) with a matching edge. By assumption, $p >0$ on $R$. Furthermore, there are no edges in $E^=$ joining $R$ to $J\backslash M_R$. Hence decreasing $p$ on $R$ and increasing $q$ on $M_R$ maintains feasibility and
hence optimality of the given dual solution (note that $\sum p_i + \sum q_j$ remains constant). We therefore decrease $p$ on $R$  until eventually a new edge $e$ joining $R$ to $J\backslash M_R$ enters $E^=$, in which case we update (enlarge) the current set $R \subseteq I$ and continue. After at most $n=|I|$ steps, $j$ will get connected to an item $i$ with
$p_i=0$ as required in condition $(ii)$ and the iteration is complete. We finish our description of the algorithm by 
opbserving that, while proceeding this way (trying to connect a current $j$ as above to the zero price items), we
 never destroy any $M$-alternating paths connecting other buyers $j'$ to zero price items. Indeed if any such path $P$
 would pass through $i\in R$ as above, then also $j$ were connected to zero price items (follow the $M$-alternating path from $j$ to $i$ and then switch to $P$). Thus, indeed, after at most $n=|I|$ iterations we will eventually reach an
optimal dual solution $(p,q)$ satisfying condition $(ii)$ in Lemma \ref{buyeropt}, \emph{i.e.,} buyer optimal prices $p$.\\

In section \ref{s-VCG} we will see how buyer optimal prices arise in a natural (and at the first glance  different) way
in the context of sealed bid auctions.

\section{VCG}\label{s-VCG}
The general idea behind VCG mechanisms is that items should be assigned to buyers so as to maximize the total valuation
(``social welfare'') and that prices are determined so that each buyer $j$ pays for the ``harm he does to the others''
by taking the item (or, in more general situations, the goods) assigned to him. More precisely, let $M$ be any 
maximum valuation matching of value $v^*=v(M)$. For any $ij \in M$ compute $v_{-j}:= v^*(G\backslash j)$, the maximum 
valuation of a matching when bidder $j$ is removed or (to maintain our assumption $|I|=|J|$) replaced by a dummy
bidder with  valuation $0$ for all items.\\

 We then compare this to 
$v_{-j}^{-i}:= v^*(G\backslash \{i,j\})$, the maximum valuation if buyer $j$ has left and taken item $i$. The difference
between these two, \emph{i.e.,} $v_{-j}-v_{-j}^{-i}$ is the harm that buyer $j$ does to the others if he takes away item
$i$ and, in VCG-terminology, this is the \emph{personalized} price $p_{ij}$ of item $i$ \emph{for} buyer $j$. The VCG 
mechanism then assigns items to buyers according to the optimal matching $M$ and asks each buyer $j$ to pay his
price $p^{VCG}_i:=p_{ij}$ for the item $i$ he is assigned to.

We claim that the VCG prices $p^{VCG}_i$ defined this way are exactly the buyer optimal market clearing prices (\emph{cf.}
also section 15.9 in \cite{EK10}).

\begin{lemma}\label{buyeropt=VCG}
Let $p$ be buyer optimal market clearing prices and $M$ be any maximum valuation matching. Then
$p_i=v_{-j}-v_{-j}^{-i}$ holds for all $ij \in M$.
\end{lemma}
\begin{proof}
 As $M$ is a maximum valuation matching, so is $M\backslash\{i,j\}$ for the
subgraph $G\backslash \{i,j\}$ for any $ij \in M$. In other words, $v_{-j}^{-i}=v^*-v_{ij}$ for $ij\in M$.\\

Secondly, observe what happens if we replace $j$ by a dummy bidder: Let  $q \in \R^n$ correspond to $p$, so that
 $v^*=\sum p_i+\sum q_j$. By Lemma \ref{buyeropt}, there exists an $M$-alternating path $P$ from $j$ to some zero priced
item $i$.

Construct a new matching $\bar{M}$ by  switching $M$ along $P$ (thereby unmatching $i$ and $j$) and matching $i$ and $j$ with each other. Correspondingly, obtain $\bar{q}$ from $q$ by decreasing $q_j$ to $0$. Then  $\bar{M}$  is (still) complementary with $(p,\bar{q})$  and hence optimal. Thus  $v_{-j}=v^*-q_j$. Consequently, the VCG price that $j$
has to pay for the item $i$ that is assigned to him equals 
$$p^{VCG}_i=v^*-q_j -(v^*-v_{ij})=v_{ij}-q_j=p_i$$
as claimed. \qed.
\end{proof}

VCG prices have the desirable property of being \emph{incentive compatible}: Buyers cannot profit by lying 
about their true valuations $v_{ij}$. Hence this is true for the mechanism thatcomputes (and posts)  buyer optimal
 market clearing prices (and assigns items to buyers according to a maximum value perfect matching $M$). We include
a proof for convenience of the reader.

\begin{theorem}\label{incentive}
Assigning items to buyers at (posted) buyer optimal market clearing prices is incentive compatible.
\end{theorem}
\begin{proof}
Assume that buyer $j$, when reporting his valuations $v_{ij}$ truthfully gets assigned to item $i$ at a price of $p_i=
v_{-j}-v_{-j}^{-i}$ (according to Lemma \ref{buyeropt=VCG}). We compare his profit $q_j=v_{ij}-p_i$ to the outcome in
 case he reports valuations $\tilde{v}_{ij}, i \in I$, and, as a consequence, gets assigned to item $\tilde{i}\in I$ 
(possibly different from $i$). The resulting price is then (again according to Lemma \ref{buyeropt=VCG}) $\tilde{p}_i=
\tilde{v}_{-j}-\tilde{v}_{-j}^{-i}$, where the tilde indicates that the corresponding values are computed based on valuations $\tilde{v}_{ij}, i\in I$ of player $j$.\\

Let $M$ and $\tilde{M}$ denote the corresponding assignments computed \emph{w.r.t.} given valuations $v$ and $\tilde{v}$, \emph{resp.} Thus $ij\in M$ and $\tilde{i}j \in \tilde{M}$. As $M$ is a maximum value matching \emph{w.r.t.} $v$, we
have $v(M) \ge v(\tilde{M})$, \emph{i.e.,}
$$ v_{ij}+v_{-j}^{-i} \ge v_{\tilde{i}j}+v_{-j}^{-\tilde{i}}=v_{\tilde{i}j}+\tilde{v}_{-j}^{-\tilde{i}}.$$
(The last equation is due to the fact that $v_{-j}^{-\tilde{i}}$ does not depend on the valuations of $j$.) Hence the
price that $j$ were to pay for item $\tilde{i}$ (when he misreports his valuations as $\tilde{v}_{ij}$) is, according
to Lemma \ref{buyeropt=VCG},
$$p_{\tilde{i}}= \tilde{v}_{-j}-\tilde{v}_{-j}^{-\tilde{i}}=v_j-\tilde{v}_{-j}^{-\tilde{i}}\ge v_{-j}-(v_{ij}+v_{-j}^{-i}
-v_{\tilde{i}j}) =p_i-v_{ij}+v_{\tilde{i}j},$$
showing that $v_{\tilde{i}j}-\tilde{p}_i \le v_{ij}-p_i$. So misreporting his valuations cannot increase his profit. \qed
\end{proof}

\bibliographystyle{plain}

\end{document}